\documentclass[12pt]{article}
\usepackage{graphics}
\usepackage{amsmath}
\usepackage{amssymb}
\usepackage{epsfig,wrapfig}
\usepackage{graphicx}
\usepackage{amsmath}
\usepackage{epsfig,wrapfig}
\usepackage{fullpage}
\usepackage{color}

\usepackage{cite}

\newtheorem{theorem}{Theorem}

\title{\LARGE \bf Sign Patterns of Inverse Doubly-Nonnegative Matrices }

\author{Sandip Roy and Mengran Xue}

\begin{document}
\maketitle

\begin{abstract}
The sign patterns of inverse doubly-nonnegative matrices are examined.  A necessary and sufficient condition is developed for a sign matrix to correspond to an inverse doubly-nonnegative matrix.  In addition, for a doubly-nonnegative matrix whose graph is a tree, the inverse is shown to have a unique sign pattern, which can be expressed in terms of a two-coloring of the graph. \newline
\end{abstract}

\noindent {\bf Keywords and Subject Classifications:} positive matrices and their generalizations (15B48), sign pattern matrices (15B35), Hermitian matrices (15B57)

\section{Introduction and Problem Formulation}

There is a considerable literature on inverse nonnegative matrices, which characterizes their sign patterns \cite{sign1,sign2,sign3}, positivity or spectra of certain splittings \cite{splitting1,splitting2}, and principal submatrices \cite{goo1}, among other properties (e.g. \cite{goo2}). In several application domains, inverses of nonnegative matrices which additionally are symmetric and positive semi-definite -- known as 
doubly nonnegative matrices -- are of interest.    For instance, the controllability analysis of network synchronization processes involves the inversion of Gram matrices whose entries are also nonnegative or strictly positive \cite{controllab1,controllab2}.  Analogously, doubly-nonnegative matrices and their inverses may arise in the context of e.g. inverse covariance estimation for linear regression \cite{covar}, quadratic programming \cite{quadratic}, and analysis of social/biological networks with antagonistic interactions \cite{antagonistic}.  In several of these settings, the sign pattern of the inverse matrix is particularly important as it gives insight into the co-dependencies among quantities of interest, and/or allows majorization of metrics.  Despite this broad motivation, to the best of our knowledge only one recent study, by M.~Fiedler in 2015, has considered the sign patterns of inverse doubly-nonnegative matrices \cite{fiedlerdouble}.  Fiedler's study characterizes the inverses of entry-wise strictly positive and positive definite matrices, obtaining necessary conditions on the number of negative entries, as well as sufficient conditions for a sign pattern to correspond to an inverse matrix.

Here, we examine the sign patterns of inverse doubly-nonnegative matrices, first for the general case and then for the specialization that the doubly-nonnegative matrix has a tree structure. For the general case, a necessary and sufficient condition for a given sign pattern to correspond to a inverse doubly-nonnegative matrix is developed; the result generalizes that in \cite{fiedlerdouble} to encompass nonnegative as well as strictly-positive matrices, and to achieve a necessary and sufficient condition.  For tree-structured doubly-nonnegative matrices, the sign pattern of the inverse is characterized completely in terms of the graph topology.

An $n \times n$ real square matrix $A=[a_{ij}]$ is considered.  The matrix $A$ is assumed to be doubly nonnegative, i.e. entry-wise nonnegative and symmetric positive semi-definite.  We make two additional assumptions throughout the study: 1) $A$ is invertible (and hence positive definite), and 2) $A$ is irreducible.  Invertibility is assumed since our interest is in characterizing the inverse, and irreducibility can be assumed without loss of generality since otherwise the inverse can be characterized for each diagonal block.  A graph $\Gamma=(V,E)$ is defined, where $V=\{ 1,2,\hdots , n \}$, and $(i,j) \in E$ for distinct vertices $i$ and $j$ if and only if $a_{ij}>0$.  Irreducibility implies that $\Gamma$ is connected.  We note that doubly nonnegative matrices have been studied in the linear algebra and also optimization literatures \cite{doubnon0,doubnon1,doubnon2,doubnon3}, with a particular focus on relating them with the classes of completely positive and copositive matrices, and exploiting their structure to solve quadratic optimization problems.

Our goal is to characterize the sign pattern of the entries in $A^{-1}$.  For this analysis, we define a sign matrix $S$ as one whose entries are either negative signs
(denoted by --) or positive signs (denoted by +).  For any real matrix $Q=[q_{ij}]$,
the sign matrix $S(Q)$ for the matrix is 
formed by replacing the negative entries in $Q$  ($q_{ij}<0$) with a negative sign, and
the nonnegative entries ($q_{ij} \ge 0$) with a positive sign.  Also, for any $n \times n$ square sign matrix $S=[s_{ij}]$ which is symmetric ($s_{ij}=s_{ji}$), a graph $\Delta(S) =(V_d,E_d)$ is defined as follows: $V_d=\{ 1, 2, \hdots, n \}$, and  $(i,j) \in E_d$ for distinct
vertices $i$ and $j$ if $s_{ij}$ is a negative sign.  We refer to $\Delta(S)$ as the {\em negative-sign graph} of the sign matrix $S$. We notice that the
negative-sign graph specifies the negative off-diagonal entries in the sign matrix $S$.

Two questions are studied. 1) Is a specified sign matrix $\overline{S}$ {\em feasible}, i.e. is it possible for the sign matrix of an inverse doubly nonnegative matrix to equal $\overline{S}$? 2) Can the sign matrix of the inverse be exactly determined given the graph $\Gamma$ of a doubly-nonnegative matrix?

\section{Results}

 A necessary and sufficient condition for the feasibility of a sign matrix is given in the following theorem:
 \begin{theorem}
An $n \times n$ sign matrix $\overline{S}$ is feasible, i.e. there is an irreducible doubly-nonnegative matrix $A$ such that $\overline{S}=S(A^{-1})$, if and only if: 1) $\overline{S}$ is symmetric, 
3) the diagonal entries in $\overline{S}$ are +, and 2) the negative-sign graph $\Delta(\overline{S})$ is connected.
 \end{theorem}
 
 \begin{proof}
 Since $A^{-1}$ is symmetric and positive definite for any doubly nonnegatve matrix $A$, $\overline{S}$ can be a sign matrix of $A^{-1}$ only if it is symmetric and its diagonal entries are +.
 
We prove that it is necessary for the negative sign graph $\Delta(\overline{S})$ to be connected by contradiction.  Thus, assume that $\Delta(\overline{S})$ is not connected, but
$A^{-1}$ has the sign pattern $\overline{S}$ (i.e., $\overline{S}=S(A^{-1})$.
Then there exists a permutation matrix $P$ such
that $Z=PA^{-1} P^{-1}$ can be
partitioned as $Z=\begin{bmatrix} Z_{11} & Z_{12} \\ Z_{12}^T & Z_{22} \end{bmatrix}$, where $Z_{11}$ and $Z_{22}$ are square and positive semidefinite, and $Z_{12}$ is elementwise nonnegative.  

Now consider the spectrum of the matrix A.  The matrix $A$ is nonnegative and irreducible by assumption, and further $A$ has strictly positive diagonal entries since it is assumed to be positive definite.  It thus follows that $A$ is also aperiodic.  From the Frobenius-Perron theory, it follows that $A$ has a real positive eigenvalue $\lambda$ with algebraic multiplicity $1$ which is dominant (strictly larger in magnitude than all other eigenvalues).  Further, $A$ has an eigenvector ${\bf v}$ associated with $\lambda$ which is strictly positive, and unique to within a scaling.   From similarity, notice
that $\lambda$ is also the simple dominant eigenvalue of $Y=PAP^{-1}$, with corresponding strictly positive eigenvector given by ${\bf x}=P{\bf v}$.
Thus, the matrix $Z=Y^{-1}$ has a minimum eigenvalue $\gamma=\frac{1}{\lambda}$ which has algebraic multiplicity $1$, with corresponding strictly positive eigenvector ${\bf x}$.

To continue, notice that  
$\gamma=\frac{{\bf x}^T Z {\bf x}}{{\bf x}^T {\bf x}}$.  Substituting the partitioned form of $Z$, we 
obtain that 
\begin{equation}
\gamma=\frac{\begin{bmatrix} {\bf x}_1^T & {\bf x}_2^T \end{bmatrix} \begin{bmatrix} Z_{11} & Z_{12} \\ Z_{12}^T & Z_{22} \end{bmatrix} \begin{bmatrix}
{\bf x}_1 \\ {\bf x}_2 \end{bmatrix}}{\begin{bmatrix} {\bf x}_1^T & {\bf x}_2^T \end{bmatrix} \begin{bmatrix}
{\bf x}_1 \\ {\bf x}_2 \end{bmatrix}}
=\frac{{\bf x}_1^T Z_{11} {\bf x}_1+{\bf x}_2^T Z_{22} {\bf x}_2+2{\bf x}_1^T Z_{12} {\bf x}_2}{{\bf x}_1^T {\bf x}_1+{\bf x}_2^T {\bf x}_2},
\end{equation}
where ${\bf x}_1$ and ${\bf x}_2$ have commensurate dimensions with the partitions of $Z$.  From the Cauchy interlacing theorem, we see that the smallest eigenvalues of $Z_{11}$ and $Z_{22}$ are each at least $\gamma$, since they are principal submatrices of the positive-definite matrix $Z$.  From the Courant-Fisher theorem for symmetric matrices, it thus follows that ${\bf x}_1^T Z_{11}{\bf x}_1 \ge 
\gamma {\bf x}_1^T {\bf x}_1$ and ${\bf x}_2^T Z_{22}{\bf x}_2 \ge 
\gamma {\bf x}_2^T {\bf x}_2$.  Substituting,
we get:
\begin{equation}
    \gamma \ge 
    \frac{\gamma {\bf x}_1^T {\bf x}_1 +\gamma {\bf x}_2^T {\bf x}_2 +2 {\bf x}_1^T Z_{12} {\bf x}_2}{{\bf x}_1^T {\bf x}_1+{\bf x}_2^T {\bf x}_2}, 
\end{equation}
which then implies that 
\begin{equation}
{\bf x}_1^T Z_{12} {\bf x}_2 \le 0 \label{ineq:main}
\end{equation}
Since $Z_{12}$ is nonnegative and ${\bf x}_1$
and ${\bf x}_2$ are strictly positive, 
the inequality (\ref{ineq:main}) can only hold
if $Z_{12}$ is identically $0$.  However, in this case, $Y=Z^{-1}$  is not irreducible.  Thus, $A$ also is not irreducible, and a contradiction is reached.
Thus, necessity of the conditions in the
theorem has been verified.

To prove sufficiency, suppose that the $n \times n$ matrix $\overline{S}$ satisfies the three conditions, and construct the matrix $Q$ such that $\overline{S}=S(Q)$ as follows. First, we choose all diagonal
entries $q_{ii}$, $i=1,\hdots, n$ to equal 
$n$.  Meanwhile, each off-diagonal entry $q_{ij}$ is set to $-1$ if the
entry at row $i$ and column $j$ of $\overline{S}$
is --, and is set to $0$ if the entry of $\overline{S}$ is +.   The matrix $Q$ constructed in this way is a nonsingular symmetric $M$-matrix, since it is strictly diagonally dominant with positive diagonal entries and nonpositive off-diagonal entries.  Further, from the assumption that $\Delta(\overline{S})$ is connected, it follows that $Q$ is irreducible.  Thus,  the inverse
of $Q$ is entrywise strictly positive, positive definite, and symmetric. It thus follows that $Q$ is an (irreducible) doubly nonnegative matrix. $\Box$ \newline

\end{proof}

The feasible sign patterns of inverse doubly nonnegative matrices, as characterized in Theorem 1, are restricted compared to the feasible sign patterns for inverse nonnegative matrices as developed in \cite{sign1,sign2,sign3}, even for symmetric patterns and matrices.  The following is an example of a sign matrix which is feasible for an inverse nonnegative matrix, but not for an (irreducible) inverse doubly-nonnegative matrix:
\begin{equation}
\overline{S}=
    \begin{bmatrix}
    + & - & + & + \\ - & + & + & + \\
    + & + & + & - \\ + & + & - & + 
    \end{bmatrix}
\end{equation}
The sign matrix $\overline{S}$ does not have a 
connected negative-sign graph, hence it does not meet the conditions in Theorem $1$.  However,
a matrix with this sign pattern can be inverse nonnegative. For instance, the following matrix can be checked to be inverse nonnegative (indeed inverse positive):
\begin{equation}
Q=\begin{bmatrix} 1 & -2 & 1.1 & 0.01 \\ -2 & 1 & 0.01 & 1.1 \\ 1.1 & 0.01 & 1 & -2 \\ 0.01 & 1.1 &-2 & 1 \end{bmatrix}.
\end{equation}
Of note, the above inverse-nonnegative matrix $Q$ is symmetric, however it is not positive definite and hence is not doubly nonnegative.

{\em Remark 1:}  Theorem 1 has a close connection to the elegant sign-pattern analysis of inverse entry-wise- positive and positive-definite- matrices in \cite{fiedlerdouble}, which we recently became aware of. 
Importantly,  \cite{fiedlerdouble} also recognized the central role of the graph of the negative entries in the inverse matrix, and used an appropriate permutation to verify the result.  Relative to \cite{fiedlerdouble}, our study encompasses the possibility for zero entries in both the doubly-nonnegative matrix and its inverse, and hence also achieves a necessary and sufficient condition.  This broader result is achieved using a direct analysis of quadratic forms of the inverse, rather than a Hadamard-product-based argument. \newline

Finally, we study the question of whether the graph of a doubly-nonnegative matrix uniquely specifies or determines the sign pattern of its inverse.  The notion that the inverse is uniquely determined by the graph can be formalized as follows: the graph $\Gamma$ of a doubly-nonnegative matrix is said to uniquely determine the sign pattern of the inverse if, for every doubly-nonnegative matrix $A$ with graph $\Gamma$, the sign pattern of the inverse $S(A^{-1})$ is identical.  The question of whether the graph uniquely determines the inverse is interesting from an application standpoint, as it shows whether the interdependencies captured by the inverse of a doubly-nonnegative matrix (e.g. an inverse Gramian for a nonnegative system) are structural in nature.

The following theorem shows that the sign pattern of the inverse is uniquely specified if $\Gamma$ is a tree (a connected graph with a unique path between any two vertices), and 
fully characterizes the sign pattern of the inverse in this case.  This analysis uses a two-coloring of the graph $\Gamma$, i.e. a labeling of the vertices of the graph with two colors such that
no two adjacent vertices have the same color. If $\Gamma$ is a connected tree, then it has a two-coloring, and further the two-coloring is unique except for the possibility of the reversal of all vertices' colors.  The sign matrix of  $A^{-1}$ can be specified in terms of the two-coloring, as follows:
\begin{theorem}
The graph $\Gamma$ of an irreducible doubly-nonnegative matrix uniquely determines the sign pattern of the inverse if $\Gamma$ is a tree. Specifically, for a matrix $A$ with tree graph $\Gamma$,
all diagonal entries of $\overline{S}=S(A^{-1})$ are +.  Meanwhile, the off-diagonal entry at row $i$ and column $j$ is -- if vertices $i$ and $j$ have different colors in the two-coloring of the graph; otherwise, the entry is +.
\end{theorem}

\begin{proof}
The result is proved by induction on the number of vertices.  
As a basis step, consider any $2 \times 2$ 
irreducible doubly-nonnegative matrix $A_2$; note that the graph for the matrix $A_2$, which we call $\Gamma_2$, is necessarily a connected tree.  From the matrix inversion formula for $2 \times 2$ matrices and the positivity of the determinant, it is immediate that the sign matrix $S(A^{-1})$ is $\overline{S}=\begin{bmatrix} + & - \\
- & + \end{bmatrix}$, hence the basis is verified.

Now assume that the result holds for
any $k \times k$ irreducible doubly-nonnegative matrix $A_{k}$ whose graph $\Gamma$ is a tree, i.e. the sign pattern of the inverse is commensurate with the two-coloring as specified in the theorem statement.

Consider any
$(k+1) \times (k+1)$ irreducible doubly nonnegative matrix
$A_{k+1}$ whose graph is a tree. Notice that any such matrix $A_{k+1}$
can be written as:
\begin{equation}
    A_{k+1}=\begin{bmatrix} A_k & c{\bf e}_i \\ c{\bf e}_i^T & d \end{bmatrix},
\end{equation} 
where ${\bf e}_1$ is a 0--1 indicator vector whose $i$th entry is $1$,  $c$ and $d$ are positive scalars, and $A_k$ is some doubly-nonnegative matrix whose graph $\Gamma_k$ is a connected tree. Further, the graph $\Gamma_{k+1}$ for $A_{k+1}$ is seen to be formed from graph $\Gamma_k$ through the addition of a single vertex labeled $k+1$, which connects to vertex $i$. We note that the two-coloring of vertices $1,\hdots, k$ in $\Gamma_{k+1}$ is identical to that of $\Gamma_k$, while vertex $k+1$ is different in color from vertex $i$.

Since $A_{k+1}$ is positive definite, the diagonal entries of $A_{k+1}^{-1}$ are positive.  To characterize the off-diagonal entries, we apply the block matrix inversion formula.  Doing so,
$A_{k+1}^{-1}$ is given by:
\begin{equation}
\begin{bmatrix}
(A_k-\frac{c^2}{d}{\bf e}_i{\bf e}_i^T)^{-1}
& - \frac{c}{d}(A_k-\frac{c^2}{d}{\bf e}_i{\bf e}_i^T)^{-1} {\bf e}_i  \\
-\frac{c}{d}{\bf e}_i^T (A_k-\frac{c^2}{d}{\bf e}_i{\bf e}_i^T)^{-1} & \frac{1}{d-c^2 {\bf e}_i^T A_k^{-1} {\bf e}_i } .
\end{bmatrix}
\end{equation}

Next, consider the matrix $(A_k-\frac{c^2}{d}{\bf e}_i{\bf e}_i^T)^{-1}$. 
This matrix is positive definite since it is a principal submatrix of $A_{k+1}^{-1}$.  Thus, 
the matrix $A_k-\frac{c^2}{d}{\bf e}_i{\bf e}_i^T$ is also positive definite.  Since in addition this matrix differs from $A_k$ in only one diagonal entry, it is also entry-wise nonnegative, and its graph is a tree. From the induction hypothesis, 
the matrix $(A_k-\frac{c^2}{d}{\bf e}_i{\bf e}_i^T)^{-1}$ is thus seen to have a sign pattern that is identical to that of $A_{k}^{-1}$, and commensurate with the two-coloring of the graph $\Gamma_k$ of $A_{k}$.  Using this observation, the sign of the entry of $A_{k+1}^{-1}$ at row $i$ and column $j$ can be characterized for $i=1,\hdots, k$, $j=1,\hdots, k$, and $i \neq j$.  The sign is -- if the two-coloring of $\Gamma_k$, and hence $\Gamma_{k+1}$, has different colors at vertex $i$ and $j$.  It is + otherwise.

It remains to characterize the signs of the entries in the rightmost column of $A_{k+1}^{-1}$ (specifically, the entries at rows $j=1,\hdots, k$ and column $k+1$).  From the block inverse formula, these entries are given by ${\bf h}=- \frac{c}{d}(A_k-\frac{c^2}{d}{\bf e}_i{\bf e}_i^T)^{-1} {\bf e}_i$.  Since $c$ and $d$ are positive, the column vector ${\bf h}$ 
is equal to the $i$th column of $(A_k-\frac{c^2}{d}{\bf e}_i{\bf e}_i^T)^{-1}$, multiplied by a negative scalar.  Thus, notice that the entry
at row $j$ and column $k+1$ of $A_{k+1}^{-1}$
has the opposite sign as the entry at row $j$ and column $i$.  Since the vertex $k+1$ is different in color from the vertex $i$ in the two-coloring of $\Gamma_{k+1}$, the entry at row $j$ and column $k+1$ is negative if and only if vertices $j$ and $k+1$ have opposite colors, and is positive otherwise. Thus, the sign pattern of $A_{k+1}$ has been verified, and the theorem has been proved by induction. $\Box$ \newline
\end{proof}

Theorem $2$ demonstrates that the sign pattern of the inverse of a doubly nonnegative matrix is uniquely determined, when the graph $\Gamma$ of the matrix is a tree.   We do not have a complete answer to the question of whether the graph uniquely determine the sign pattern of the inverse, in the case where it is not a tree. However, we give a rough sketch for an argument, which can be used to show that the sign pattern of the inverse is not unique for many graphs. Specifically, let us consider doubly nonnegative matrices $A$ with a particular graph $\Gamma$ that is not a tree.  Since $\Gamma$ is not a tree, two different spanning trees of $\Gamma$ can be found, say $\Gamma_1$ and $\Gamma_2$.  Then, doubly nonnegative matrices with graphs $\Gamma_1$ and $\Gamma_2$ can be constructed, say $A_1$ and $A_2$.  Since $\Gamma_1$ and $\Gamma_2$ are trees, the sign patterns of the inverses of $A_1$ and $A_2$ are indicated by Theorem 2.  Indeed, provided that the two graphs have different two colorings, the two inverses have different sign patterns per the theorem.  Then, via small additive perturbations of $A_1$ and $A_2$, two doubly nonnegative matrices with graph $\Gamma$ can be constructed, whose inverses have different sign patterns. 
Thus, the graph of a doubly nonnegative matrix does not uniquely determine the inverse, provided that it has two spanning trees with different two-colorings.  This can be seen to include all graphs that are not bipartite.    
Some technical detail is needed to design the small perturbations so that      the sign patterns of the inverses are guaranteed and double nonnegativity is maintained; we have not developed this analysis in detail, and hence do not provide a formal result or proof.        We note that similar questions regarding determination of sign patterns have been studied in the context of sign solvability of linear systems \cite{brualdi}, however to the best of our knowledge these studies do not address the doubly-nonnegative matrix case.

{\em Remark 2:} The proof of Theorem 2 provides some further characterizations of the entries in the inverse matrix beyond their sign pattern.  For instance, the block inversion formula clarifies that, for rows of the inverse matrix corresponding to leaf vertices, the off-diagonal entries are a fixed multiple of the corresponding entries of the row corresponding to the parent node.  

{\em Remark 3:} Theorem 2 may equivalently be phrased in terms of the distances between vertices in the graph $\Gamma$: the entry at row $i$ and column $j$ of $A^{-1}$ is negative if and only if the distance between vertices $i$ and $j$ in $\Gamma$ is odd.

\subsection*{Acknowledgements}

This work was partially supported by United States National Science Foundation grants CNS-1545104 and CMMI-1635184.

\bibliographystyle{apa}

\end{document}